\documentclass[pra,aps,eqsecnum,superscriptaddress,nofootinbib,notitlepage,longbibliography,draft]{revtex4-2}

\usepackage[a4paper, left=2cm, right=2cm, top=2cm, bottom=2cm]{geometry}
\usepackage{setspace}
\usepackage[english]{babel}
\addto\extrasenglish{%
}
\usepackage[utf8]{inputenc}
\usepackage[T1]{fontenc}
\usepackage{lmodern}
\usepackage{amsmath,amsthm,mathdots,amssymb,bm,bbm,mathrsfs,mathtools}
\usepackage[final]{graphicx} 
\usepackage{subcaption}
\usepackage{tikz,pgfplots}\pgfplotsset{compat=1.16}

\usepackage{thm-restate}

\newtheorem{lem}{Lemma}[section]

\newtheorem{fact}{Fact}[section]

\newtheorem{alg}{Algorithm}[section]

\usepackage{comment}

\usepackage[colorlinks=true, urlcolor=violet, linkcolor=blue, citecolor=red, hyperindex=true, linktocpage=true, pagebackref=true, draft=false]{hyperref}
\usepackage{mleftright}\mleftright 

\usepackage{enumitem}
\mathtoolsset{showonlyrefs=true}
\usepackage{microtype,color,graphicx} 
\usepackage{tikz-cd}
\usepackage[table]{xcolor}
\usepackage{array}
\usepackage{multirow}
\usepackage{booktabs}
\usepackage{tablefootnote}

\renewcommand{\vec}{\bm}

\newcommand{\BC}{\mathbb{C}}

\newcommand{\CO}{\mathcal{O}}

\newcommand{\BR}{\mathbb{R}}

\newcommand{\vA}{\bm{A}}

\newcommand{\vB}{\bm{B}}

\newcommand{\vh}{\bm{h}}

\newcommand{\vH}{\bm{H}}
\newcommand{\vI}{\bm{I}}

\newcommand{\vV}{\bm{V}}

\newcommand{\vX}{\bm{X}}
\newcommand{\vY }{\bm{Y }}
\newcommand{\vZ}{\bm{Z}}

\renewcommand{\L}{\left}
\newcommand{\R}{\right}

\newcommand{\vertiii}[1]{{\left\vert\kern-0.25ex\left\vert\kern-0.25ex\left\vert #1 \right\vert\kern-0.25ex\right\vert\kern-0.25ex\right\vert}}
\newcommand{\norm}[1]{\Vert {#1} \Vert}

\newcommand{\labs}[1]{\left\vert {#1} \right\vert}
\newcommand{\lnorm}[1]{\left\Vert {#1} \right\Vert}
\newcommand{\e}{\mathrm{e}}

\newcommand{\ri}{\mathrm{i}}
\newcommand{\rd}{\mathrm{d}}

\newcommand*{\poly}{\mathrm{Poly}}

\DeclarePairedDelimiterX{\braket}[1]{\langle}{\rangle}{#1}
\DeclarePairedDelimiterX\ketbra[2]{| }{|}{#1 \delimsize\rangle\!\delimsize\langle #2}	
\DeclarePairedDelimiterX\dotp[2]{\langle}{\rangle}{#1, #2}

\DeclareMathAlphabet{\dutchcal}{U}{dutchcal}{m}{n}
\SetMathAlphabet{\dutchcal}{bold}{U}{dutchcal}{b}{n}
\DeclareMathAlphabet{\dutchbcal} {U}{dutchcal}{b}{n}

\makeatletter
\DeclareRobustCommand*{\pmzerodot}{%
	\nfss@text{%
		\sbox0{$\vcenter{}$}
		\sbox2{0}%
		\sbox4{0\/}%
		\ooalign{%
			0\cr
			\hidewidth
			\kern\dimexpr\wd4-\wd2\relax 
			\raise\dimexpr(\ht2-\dp2)/2-\ht0\relax\hbox{%
				\if b\expandafter\@car\f@series\@nil\relax
				\mathversion{bold}%
				\fi
				$\cdot\m@th$%
			}%
			\hidewidth
			\cr
			\vphantom{0}
		}%
	}%
}

\usepackage{ifdraft}
\ifdraft{
	\newcommand{\authnote}[3]{{\color{#3} {\bf  #1:} #2}}	
}{
	\newcommand{\authnote}[3]{}
}

\usetikzlibrary{quantikz}
\begin{document}

\title{Catalytic Tomography of Ground States}

\author{Chi-Fang Chen}
\email{achifchen@gmail.com}
\affiliation{University of California, Berkeley, CA, USA}
\author{Robbie King}
\email{robbieking@google.com}
\affiliation{Google Quantum AI, Santa Barbara, CA, USA}
\affiliation{University of California, Berkeley, CA, USA}

\begin{abstract}
We introduce a simple protocol for measuring properties of a gapped ground state with essentially no disturbance to the state. The required Hamiltonian evolution time scales inversely with the spectral gap and target precision (up to logarithmic factors), which is optimal. For local observables on geometrically local systems, the protocol only requires Hamiltonian evolution on a quasi-local patch of inverse-gap radius. Our results show that gapped ground states are algorithmically readable from a single copy without a recovery or rewinding procedure, which may drastically reduce tomography overhead in certain quantum simulation tasks.
\end{abstract}

\maketitle

\section{Introduction}

A notorious feature of quantum mechanics is that measurements are generally destructive. The Born rule asserts that measurement collapses the quantum state into the eigenbasis of the observable; therefore, learning the expectation of an observable generally requires multiple copies of identical, unknown quantum states. Together with the uncertainty principle and the no-cloning theorem, these spooky properties of quantum mechanics render tomography of quantum states a rich and active subject, e.g., in full state tomography~\cite{haah2016sample,o2016efficient}, achieving high precision~\cite{knill2007optimal,huggins2022nearly}, measuring many observables (`shadow tomography')~\cite{aaronson2018shadow,huang2020predicting,huang2021information}, and measuring fermionic observables~\cite{bonet2020nearly,zhao2021fermionic,wan2023matchgate,king2025triply}.

In the quantum simulation context, the tomography problem naturally appears as an indispensable subroutine: given the ground state of a many-body Hamiltonian, we ultimately care about its measurable, physical properties. However, generic measurement schemes require the number of copies to grow polynomially with the precision. Therefore, in search for quantum speedup in many-body simulation, the \textit{multiplicative} nature of tomography costs has gradually emerged as a significant burden, especially when preparing each copy of the ground state is already resource-demanding~\cite{huggins2021efficient}. By contrast, classical methods such as tensor networks~\cite{schollwock2011density,orus2019tensor,verstraete2023density}, density functional theory~\cite{gross1995density,parr1995density}, and many other techniques in quantum chemistry~\cite{szabo1996modern,helgaker2000molecular} are often free of tomography costs, as an entire description of an ansatz state is stored in classical memory, and observables of interest are readily available.

In this work, we give a simple, general tomography protocol for gapped ground states with the following simultaneous desirable properties:
\begin{itemize}
    \item (Catalytic.) The original ground state remains essentially undisturbed throughout the measurement.
    \item (Efficient.) Each measurement uses, up to logarithmic factors, a Hamiltonian evolution time scaling inversely with the spectral gap and precision, which is optimal.
    \item (Quasi-local.) For local observables in geometrically local Hamiltonians, measurements are quasi-local and require evolving the Hamiltonian only on a patch of radius inversely proportional to the spectral gap. 
\end{itemize}
Ref.~\cite{farhi2010quantum} showed that tomography of gapped ground states requires only a single copy, by restoring (by a recovery map) or rewinding (by Mariott-Watrous \cite{marriott2005quantum} or QMA-amplification \cite{nagaj2009fast}) the state after measurement. Our results offer a conceptually more direct approach: our protocol ensures that the ground state remains undisturbed throughout the entire measurement, as in running quantum phase estimation for eigenstates. This is done locally by `filtering out' the destructive part of the observable, a technique that is used to study properties of ground states and to design dissipative quantum algorithms. 

The existence of a catalytic tomography protocol shows that gapped ground states do not always behave like Schrodinger's cat, but can be algorithmically \emph{readable} in a similar way to a classical memory. Observables can be measured efficiently and locally without effectively any disturbance to the state, which may drastically reduce tomography overhead in practice.

\section{Setup}
\label{sec:setup}
Consider an $n$-qubit Hamiltonian $\vH$ with ground state $\ket{\psi_0}$ and a Hermitian observable $\vA$ such that the operator norm is bounded by $\norm{\vA}\le1$. The goal, then, is to estimate the ground state expectation $\bra{\psi_0} \vA \ket{\psi_0} \in [-1,1]$ to guaranteed precision $\epsilon$, on a digital, fault-tolerant quantum computer. For our main result, which holds in the fully general case, we make three assumptions:

\begin{itemize}
    \item (Gapped, unique ground state.) A key assumption in our catalytic protocol is that the Hamiltonian has a unique ground state with a spectral gap. Consider the abstract eigendecomposition of the Hamiltonian
    \begin{align}
    \vH = \sum_{i=0}^{2^n-1} E_i \ket{\psi_i}\bra{\psi_i}\quad \text{where}\quad E_0 \le E_1\cdots\le E_{2^n-1}.
    \end{align}
    We say a Hamiltonian $\vH$ is $\Delta$-gapped if it has a unique ground state $\ket{\psi_0}\bra{\psi_0}$ such that $E_1-E_0 \ge \Delta \ge 0$.     
    A priori, our protocol only requires knowledge of a lower bound for the spectral gap (which, in practice, could simply be a guess).
    \item (Hamiltonian evolution.) Our protocol uses block-box access to controlled Hamiltonian forwards and backwards evolution $e^{\pm i\vH t}$. In fact, we will only employ (controlled) Heisenberg dynamics $e^{i\vH t} \vA e^{-i\vH t}.$
    \item (Access to $\vA$.) We assume access to a block encoding of $\vA$; this is automatically guaranteed for any self-adjoint unitary circuit such as a Pauli string.
\end{itemize}

Our second result specializes in geometrically local Hamiltonians with spin or fermionic degrees of freedom. On a lattice $\Lambda = [L]^{D}$ of $n = \labs{\Lambda}$ sites, we consider Hamiltonians $\vH$ where each term $\vh_S$ acts on at most $q$-sites 
\begin{align}
    \vH = \sum_{S\subset \Lambda} \vh_S\quad \text{where}\quad \norm{\vh_S} \le 1\label{eq:localHam}
\end{align}
such that $[\vh_S,\vh_{S'}] = 0$ if $S\cap S' =\emptyset$ and $\vh_S = 0$ if $\text{diam}(S) \ge R$ (in Manhattan distance). Such individual terms can be Pauli strings or even-degree fermionic operators supported on $S$. We will assume the geometric parameters $R$, $q$, and $D$ to be constants, while the explicit dependence can be computed in principle. We also assume that the observable $\vA$ of interest is also supported on a subset of the lattice $A \subset \Lambda$. Generally, the argument holds whenever a quantitative Lieb-Robinson bound~\cite{chen2023speed} is available, such as with power-law~\cite{kuwahara2020strictly,chen2019finite} or bounded degree interactions~\cite{hastings2006spectral,nachtergaele2006lieb}.

\subsection*{Prior work}

\begin{table}[t]
\centering
\begin{tabular}{lccc}
\toprule
\textbf{Algorithm} & \textbf{Evolution time} & \textbf{Locality radius} & \textbf{Catalytic} \\
\midrule
Quantum Zeno effect~\cite{vaidman2014protective} & $\tilde{O}(1 / (\Delta \epsilon^2 \delta))$* & Entire system & No \\
Adiabatic measurement~\cite{aharonov1993measurement} & $\tilde{O}(1 / (\Delta \epsilon \sqrt{\delta}))$* & Entire system & No \\
State restoration~\cite{farhi2010quantum} & $\tilde{O}(1 / (\Delta \epsilon^2))$ & Entire system & No \\
Marriott-Watrous~\cite{marriott2005quantum,farhi2010quantum} & $\tilde{O}(1 / (\Delta \epsilon^2))$ & Entire system & No \\
Fast QMA amplification~\cite{nagaj2009fast,farhi2010quantum} & $\tilde{O}(1 / (\Delta \epsilon))$ & Entire system & No \\
Catalytic ground state tomography (\autoref{alg:main}) & $\tilde{O}(1 / (\Delta \epsilon))$ & Entire system & Yes \\
Local catalytic ground state tomography (\autoref{alg:local}) & $\tilde{O}(1 / (\Delta \epsilon))$ & $\tilde{O}(1 / \Delta)$ & Yes \\
Lower bounds (\autoref{sec:lower}) & $\Omega(1 / (\Delta \epsilon))$ & $\Omega(1 / \Delta)$ & \\
\bottomrule
\end{tabular}
    \caption{Comparison of techniques for measuring the expectation value of an observable on a gapped ground state. The desired precision of the estimate is $\epsilon$; the spectral gap is $\Delta$; and the disturbance incurred to the state is $\delta$. The $\tilde{O}$ ignores logarithmic factors such as $\log(1/\delta)$ and $\log(1/\epsilon)$. Catalytic tomography is the only known protocol that is both quasi-local for local observables and has Heisenberg-limited scaling in $\epsilon$. *These evolution times are our completion of their arguments.}
\label{tab:prior}
\end{table}

We know that phase estimation is catalytic for an eigenstate of the observable: we can estimate the observable to high precision without disturbing the eigenstate. Therefore, if the goal is only to catalytically estimate the energy of a gapped ground state, phase estimation suffices.

The possibility of tomography from a single copy for ground states and metastable states was first explored in Refs.~\cite{aharonov1993measurement,aharonov1996adiabatic,aharonov1993meaning,vaidman2014protective} under the name `protective measurement'. One protocol involves coupling the system to a probe with a weak adiabatic interaction \cite{aharonov1993measurement,aharonov1996adiabatic}.\footnote{We thank Hsin-Yuan Huang and Daniel Ranard for inspiring discussions on the adiabatic approach.} Another exploits the quantum Zeno effect \cite{vaidman2014protective}, reminiscent of the Elitzur-Vaidman bomb tester \cite{elitzur1993quantum}. In order for these protocols to ensure the state is disturbed by at most $\delta$, the required Hamiltonian evolution time scales inverse-polynomially in $\delta$, whereas a truly catalytic protocol should have scaling polynomial in $\log(1/\delta)$.

Ref.~\cite{farhi2010quantum} achieves similar results to us by efficiently restoring the state after measurement using the Hamiltonian. In addition to their own state restoration algorithm, they achieve state restoration using the QMA amplification schemes of Ref.~\cite{marriott2005quantum} and Ref.~\cite{nagaj2009fast}, which were designed to amplify the acceptance probability of a QMA verification protocol using only a single copy of the witness state. Similar to our results, these state restoration protocols achieve the optimal Hamiltonian evolution time $\tilde{\mathcal{O}}(1/(\Delta \epsilon))$, scaling only logarithmically in the desired disturbance $\delta$.

All methods mentioned above require acting on the entire system, even when measuring a local operator. It is worth comparing the above with the \textit{local Markov property} of thermal states: local noise channel can be recovered by the natural quantum Gibbs sampling dynamics~\cite{kato2025clustering, chen2025quantum, bakshi2024high} acting quasi-locally near the noise. Even though the recovery map here also enjoys quasi-locality, the runtime depends on a local mixing time. Most importantly, the possibility of noise-recovery for a mixed state--- even for classical Markov chains Monte Carlo methods --- does not imply tomography of observables; see discussion in~\cite{bergamaschi2025structural}. It remains open to bridge the Markov property of mixed states and nondestructive tomography of ground states.

\begin{figure}[t]
    \centering
    \includegraphics[width=0.65\linewidth]{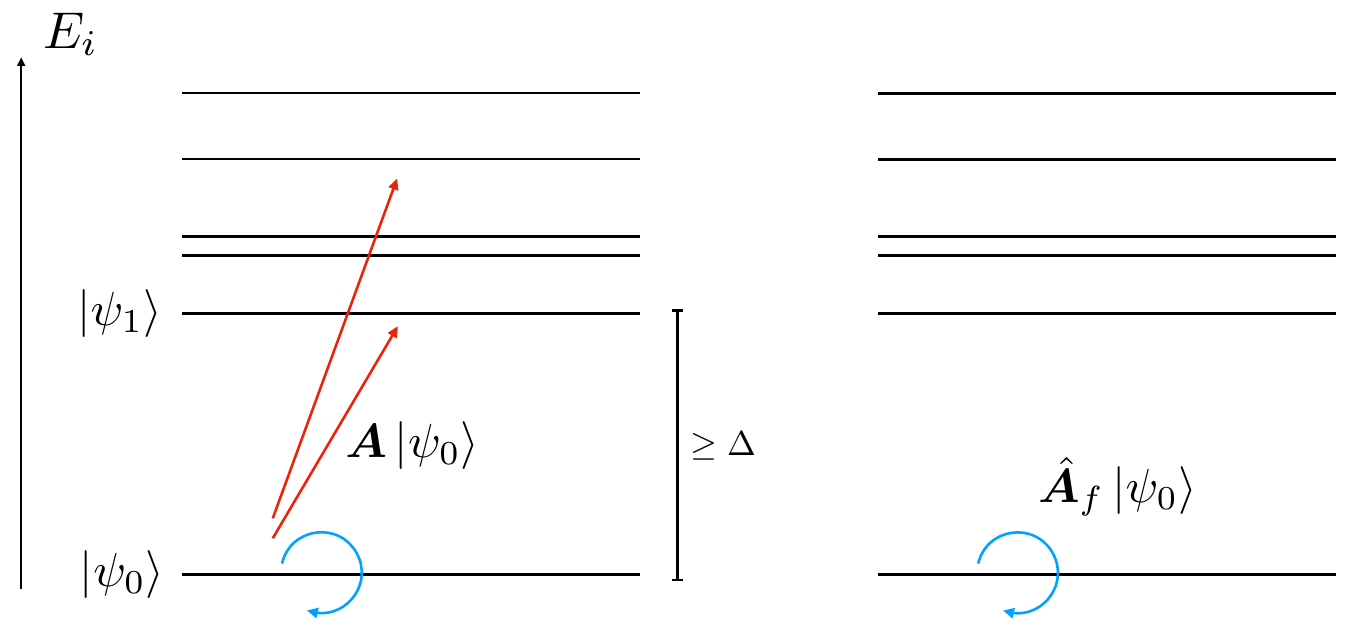}
    \captionsetup{justification=raggedright,singlelinecheck=false}
    \caption{ (Left) Acting operator $\vA$ on the ground state incurs transitions to excited states that may be difficult to undo. (Right) When we apply the filtered operator $\hat{\vA}_f,$ the transition amplitudes to the excited state are highly suppressed, and the state remains at the ground state. 
    }
    \label{fig:AAf}
\end{figure}

\section{Results}

Our catalytic tomography protocol offers a different picture to all of these previous ideas: rather than acting weakly and adiabatically, or restoring, we can algorithmically kill the destructive part of the measurement `before it even happens' (see \autoref{sec:ideas}). In addition, ours is the first protocol that becomes quasi-local for local observables on finite-dimensional lattices; see \autoref{tab:prior}.

\begin{alg}[Catalytic ground state tomography] \label{alg:main}
    Given one copy of the $\Delta-$gapped ground state $\ket{\psi_0}$ of a known Hamiltonian $\vH$ and a Hermitian observable $\norm{\vA}\le 1$, there is a measurement protocol that
    returns an $\epsilon$-approximate, $(1-\delta)$-confidence estimate of 
    \begin{align}
        \bra{\psi_0} \vA \ket{\psi_0} \quad &\text{using}\quad \CO\L(\frac{\log(1/\delta) (\log(1/\delta)+\log(1/\epsilon))}{\Delta\epsilon}\R)\quad \text{controlled Heisenberg evolution time for $\vH$} \\
        &\text{and}\quad \CO\L(\frac{\log(1/\delta)}{\epsilon}\R) \quad \text{calls to a block-encoding of $\vA$},     
    \end{align}
    and returns a state $\delta$-close to $\ket{\psi_0}$ in trace distance.
\end{alg}

See \autoref{app:implementation} for a full description and proof of \autoref{alg:main}. The total Hamiltonian evolution time comes from performing high-confidence phase estimation on an auxiliary operator $\hat{\vA}_f$, using $\CO(\log(1/\delta)/\epsilon)$ calls of $\hat{\vA}_f$, each well-approximated using Heisenberg evolution time $T = \CO( \log(1/\delta\epsilon)/\Delta).$ The scaling of $1/(\Delta\epsilon)$ is optimal in a black-box single qubit setting (\autoref{fact:lowerbound_ham}) and can be morally understood as a Heisenberg limit; we do not attempt to optimize the logarithmic factors.

\autoref{alg:main} works for any Hamiltonian $\vH$ and any observable $\vA$, given black-box access to $\vH$ and $\vA$. However, if the Hamiltonian is geometrically local and the observables are local, the same measurement protocol naturally inherits the locality from a standard Lieb-Robinson argument (e.g., \cite{hastings2006spectral, nachtergaele2006lieb,haah2020quantum,chen2023speed}), giving an improved readout protocol whose resources are independent of system size.

\begin{alg}[Local catalytic ground state tomography] \label{alg:local}
    In the setting of \autoref{alg:main}, if the Hamiltonian $\vH$ is further a local Hamiltonian~\eqref{eq:localHam} and $\vA$ is a local observable with $\norm{\vA} \leq 1$ supported on region $A\subset \Lambda$, then the Heisenberg evolution only requires Hamiltonian on a patch of radius 
    \begin{align}
        \CO\L(\frac{\log(1/\delta) + \log(1/\epsilon)}{\Delta}+ \log(\labs{A}/\delta) \R)
    \end{align}
    from region $A$. Here, $\CO(\cdot)$ suppresses geometric parameters of $\vH$ (\autoref{sec:setup}).
\end{alg}

See \autoref{app:locality} for a rigorous analysis of the locality and a proof of \autoref{alg:local}. Roughly, the truncation argument is a standard Lieb-Robinson bound, with radius scaling with the Heisenberg evolution time $T$ from \autoref{alg:main}, with some additive logarithmic correction scaling with the support size of $\vA$. Generally, the linear scaling of locality vs.~Heisenberg evolution time is optimal: free fermions have lightcone spreading at a system-size independent velocity; in terms of static correlation, the 1D TFIM exhibits an exponential correlation decay with correlation length proportional to the gap (\autoref{sec:locality_Heisenberg}). Of course, in practice, the minimal truncation radius and correlation length could depend on the particular system.

Ref.~\cite{huggins2022nearly} shows how to estimate $m$ observables to precision $\epsilon$ with $\tilde{\CO}(\sqrt{m}/\epsilon)$ queries to a state-preparation circuit using the gradient method of Ref.~\cite{gilyen2019optimizing}. In contrast, reading them one by one via catalytic ground-state tomography on a $\Delta$-gapped ground state requires Hamiltonian evolution time $\tilde{\CO}(m/(\Delta\epsilon))$. We anticipate that combining these ideas enables readout of $m$ observables on a $\Delta$-gapped ground state to precision $\epsilon$ with total evolution time $\tilde{\CO}(\sqrt{m}/(\Delta\epsilon))$.

\section{Key ideas} \label{sec:ideas}

The destructiveness of state tomography is rooted in the off-diagonal elements in the eigenbasis of the observable. When we apply the operator to a state of interest,
\begin{align}
    \vA \ket{\psi_0} = \bra{\psi_0} \vA \ket{\psi_0} \cdot \ket{\psi_0} + \ket{\psi^{\perp}},
\end{align}
the state may suffer from (unknown) leakage to the orthogonal component $\ket{\psi^{\perp}}$. Without prior knowledge of the eigenstates, it is not obvious how to algorithmically `back up' to the original ground state. Indeed, for general unknown quantum states (in fact, even for classical distributions), measuring an observable to precision $\epsilon$ requires $\sim 1/\epsilon^2$ samples (see e.g.,~\cite{scharnhorst2025optimal}). 

Our main idea allows us to perform nondestructive phase estimation of general observables on a gapped ground state, even when the observables do not commute with the Hamiltonian. Given access to the Hamiltonian $\vH$, we can efficiently create an auxiliary operator $\hat{\vA}_f$, such that
\begin{align}
   \hat{\vA}_f \ket{\psi_0}  \approx \bra{\psi_0} \vA \ket{\psi_0} \cdot  \ket{\psi_0}
\end{align}
with controllable leakage; that is, the $\hat{\vA}_f$ is approximately block-diagonal for the ground space. The construction is inspired by the filtering trick originated in the study of gapped ground states~\cite{hastings2006spectral,Hastings2007AnAL} and the design of dissipative ground and thermal state algorithms~\cite{chen2023quantum,ding2024single,chen2023efficient}.
\begin{align}
    \hat{\vA}_{f}&: = \frac{1}{\sqrt{2\pi}}\int_{-\infty}^{\infty} e^{i\vH t} \vA e^{-i\vH t} f(t) \rd t\\
     &=\sum_{i,j} \ket{\psi_i}\bra{\psi_i} \vA \ket{\psi_j}\bra{\psi_j} \hat{f}(-(E_i-E_j)).
\end{align}
Intuitively, the operator damped the transitions that increase the energy, `before they even happen'. The function is both localized in time and energy, requiring only $1/\Delta$ Heisenberg evolution time so that it can distinguish energy changes of order $\Delta$. This is an essential mechanism in the proof of correlation decay in gapped ground state~\cite{hastings2006spectral,Hastings2007AnAL}, and the construction of quasi-local, dissipative dynamics~\cite{chen2023quantum,ding2024single,chen2023efficient} for ground state or thermal states. A very natural choice would be a Gaussian with suitable variance, but other profiles are plausible. The circuit implementation for (a block-encoding of) $\hat{\vA}_f$ uses controlled Hamiltonian evolution and the preparation of a state with amplitudes given by $\sqrt{f}$ in a standard linear combination of unitaries (LCU) compilation \cite{childs2012hamiltonian}, see~\autoref{fig:A_LCU} and~\autoref{lem:implement_Af}. Performing (high-confidence) phase estimation of $\hat{\vA}_f$ on the ground state then gives the algorithm.

In the following sections, we first properly describe the filtered operators, with two plausible choices of filters. Then, we give lower bounds that essentially match the Hamiltonian evolution time and locality of our protocol.

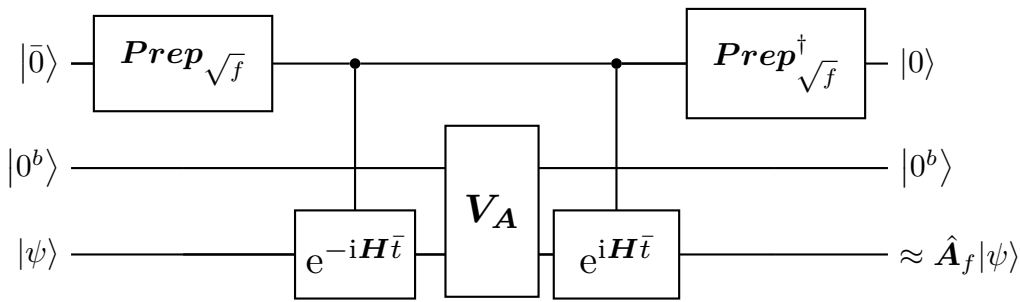
\begin{figure}[t]
\begin{center}
    \newcommand{\scalea}{1.2}
    \newcommand{\scaleb}{1.5}
    \begin{quantikz}[row sep=0.5cm] 
        \lstick{\scalebox{\scalea}{$\ket{\bar{0}}$}} & 
        \gate[style={inner xsep=2mm, inner ysep=2mm}]{\scalebox{\scalea}{$\vec{Prep}_{\sqrt{f}}$}} \qw & 
        \ctrl{2} & 
        \qw & 
        \ctrl{2} & 
        \gate[style={inner xsep=2mm, inner ysep=2mm}]{\scalebox{\scalea}{\textbf{$\vec{Prep}_{\sqrt{f}}^{\dagger}$}}}\qw & 
        \qw \rstick{\scalebox{\scalea}{$\ket{0}$}}\\
        \lstick{\scalebox{\scalea}{$\ket{0^b}$}} & 
        \qw & 
        \qw & 
        \gate[2,style={inner xsep=1mm, inner ysep=2mm}]{\scalebox{\scaleb}{$\vV_{\vA}$}} & 
        \qw & 
        \qw & 
        \rstick{\scalebox{\scalea}{$\ket{0^b}$}} \qw \\[-1mm]
        \lstick{\scalebox{\scalea}{$\ket{\psi}$}} & 
        \qw & 
        \gate[style={inner xsep=0mm, inner ysep=2mm}]{\scalebox{\scaleb}{$\e^{-\ri\vH \bar{t}}$}}\qw & 
        \qw & 
        \gate[style={inner xsep=2mm, inner ysep=2mm}]{\scalebox{\scaleb}{$\e^{\ri\vH \bar{t}}$}}\qw & 
        \qw & 
        \qw \rstick{\scalebox{\scalea}{$\approx \hat{\vA}_f\ket{\psi}$}}
     \end{quantikz}
\end{center}
\caption{Block-encoding for a Riemann sum-LCU implementation of $\hat{\vA}_f$, using state preparation circuit for $\sqrt{f}$ and a block-encoding $\vV_{\vA}$ for $\vA$; see also~the circuit for operator Fourier transform, which has an additional frequency register~\cite{chen2023quantum}. If $\vA$ is unitary, we can simply replace $\vV_{\vA}$ with $\vA$, implying $b=0$, i.e., the second registers can be omitted.}\label{fig:A_LCU}
\end{figure}

\section{Filtered operators}

To build towards the auxiliary operator $\hat{\vA}_f,$ we adapt from the notation of~\cite{chen2023quantum}. 
For any operator $\vA$ and Hamiltonian $\vH$, we can decompose according to the energy change, or  \textit{Bohr frequencies}, $\nu \in B(\vH) := \{ E_i - E_j : 0\le i,j\le 2^{n}-1\}$) by
\begin{align}
    \vA = \sum_{\nu} \vA_{\nu}\quad \text{where}\quad  \vA_{\nu} : = \sum_{E_i-E_j = \nu} \ket{\psi_i}\bra{\psi_i} \vA \ket{\psi_j}\bra{\psi_j}
\end{align}
with the symmetry property $(\vA_{\nu})^{\dagger} = (\vA^{\dagger})_{-\nu}.$
Therefore, the Heisenberg dynamics can be solved formally by
\begin{align}
\vA(t) := \e^{\ri \vH t} \vA \e^{-\ri \vH t} &= \sum_{\nu\in B} \vA_{\nu}\e^{\ri \nu t}.
\end{align}

For any integrable function $f(t):\BR\rightarrow \BC$ such that $\int \labs{f(t)} \rd t < \infty$, we can define the associated \textit{filtered operator} by an weighted integral over Heisenberg dynamics
\begin{align}
   \hat{\vA}_{f} &:= \frac{1}{\sqrt{2\pi}}\int_{-\infty}^{\infty} \vA(t) f(t) \rd t \\
   &= \sum_{\nu \in B(\vH)} \vA_{\nu} \hat{f}(-\nu) \ , \quad \text{where}\quad \hat{f}(\nu) := \frac{1}{\sqrt{2\pi}} \int_{-\infty}^{\infty} e^{-i \nu t} f(t) \rd \nu.
\end{align}
If $f(t)$ is real-valued, then $\hat{\vA}_f$ is Hermitian; the above is exactly the operator Fourier transform $\hat{\vA}_{f}(\omega) := \frac{1}{\sqrt{2\pi}}\int_{-\infty}^{\infty} e^{-i \omega t} \vA(t) f(t) \rd t$ evaluated at $\omega=0$ (e.g.,~\cite[Proposition A.2]{chen2023quantum}).

\subsubsection{Gaussian filters}

The Gaussian distribution is a natural candidate for filters for its sharp decay in both time and frequency. Consider the Gaussian weights defined by a width $\sigma >0.$
\begin{align}
\hat{f}(\nu) = e^{-\frac{\nu^2}{2\sigma^2}}\quad \text{and}\quad f(t) = \sigma e^{-\frac{\sigma^2t^2}{2}} \quad \text{such that}\quad \frac{1}{\sqrt{2\pi}}\int_{-\infty}^{\infty} f(t) \rd t = 1.\label{eq:gaussian}
\end{align}
\begin{lem}[Almost block-diagonal] Consider a Hermitian $\vA$ such that $\norm{\vA}\le 1.$ For the Gaussian function~\eqref{eq:gaussian}, the filtered operator $\hat{\vA}_f$ is almost block-diagonal in the energy basis
\begin{align}
    \lnorm{\hat{\vA}_f -\ket{\psi_0}\bra{\psi_0}\vA \ket{\psi_0}\bra{\psi_0} - \vA^{\perp}} \le e^{-\Delta^2/2\sigma^2},
\end{align}
for some operator $\vA^{\perp}$ such that $\vA^{\perp}\ket{\psi_0} = \bra{\psi_0}\vA^{\perp}= 0.$
\end{lem}
\begin{proof} The off-diagonal terms are suppressed by the Gaussian decay
    \begin{align}
        \lnorm{(\vI-\ket{\psi_0}\bra{\psi_0})\sum_{\nu \in B(\vH)} \vA_{\nu} \hat{f}(-\nu)\ket{\psi_0}\bra{\psi_0}}  &= \lnorm{ \sum_{i>0} e^{-\frac{(E_i-E_0)^2}{2\sigma^2}} \ket{\psi_i}\bra{\psi_i} \vA\ket{\psi_0}\bra{\psi_0}} \\
        &\le \lnorm{\sum_{i>0} e^{-\frac{(E_i-E_0)^2}{2\sigma^2}} \ket{\psi_i}\bra{\psi_i}} \le e^{-\Delta^2/2\sigma^2},
    \end{align}    
    and similarly for $\ket{\psi_0}\bra{\psi_0}\hat{\vA}_f(\vI-\ket{\psi_0}\bra{\psi_0})$ since $\hat{\vA}_f$ is Hermitian.
\end{proof}

\subsubsection{Exact filters}
Sometimes it may be conceptually cleaner to reason about filtered operators $\hat{\vA}_f$ with \textit{exactly} no leakage to higher energies. It is possible to consider compactly supported bump functions that only allow energy differences restricted to a prescribed interval. Of course, the actual implementation will always come with a truncation error.
\begin{lem}[Compactly supported bump functions~{\cite[Lemma 2.3]{bachmann2012automorphic}\footnote{Note the exact factors of $\sqrt{2\pi}$ due to our normalization.}}]\label{lem:compact_bump}
For any parameter $\Delta \ge 0,$ consider 
\begin{align}
    f_{\Delta}(t) = c_{\Delta} \prod_{n=1}^{\infty} \L(\frac{\sin(a_n t)}{a_n t}\R)^{2}\quad \text{where}\quad a_n = \frac{a_1}{n\ln(n)} \quad \text{for}\quad n \ge 2
\end{align}
and $a_1$ such that $\sum_{n=1}^{\infty} a_n = \Delta/2$. There is a $c_{\Delta}$ such that $\frac{1}{\sqrt{2\pi}}\int_{-\infty}^{\infty} f(t) \rd t =1,$ and 
\begin{align}
    0 \le f_{\Delta}(t) \le 2\sqrt{2\pi}(e\Delta)^2 t \exp\L(-\frac{2}{7}\frac{\Delta t}{\ln(\Delta t)^2}\R)\quad \text{where}\quad t \ge \frac{e^{1/\sqrt{2}}}{\Delta},
\end{align}
such that 
\begin{align}\label{eq:compact_f}
\hat{f}(\nu) = \begin{cases}
    0 \quad &\text{if}\quad \nu \notin [-\Delta,\Delta]\\
1\quad &\text{if} \quad \nu =0
\end{cases}.
\end{align}    
\end{lem}

\begin{lem}[Exactly block-diagonal]
Given a $\Delta$-gapped Hamiltonian, the filtered operator $\hat{\vA}_f$ for the bump function (\autoref{lem:compact_bump}) is Hermitian and block-diagonal in the energy basis
\begin{align}
    \hat{\vA}_f = \ket{\psi_0}\bra{\psi_0}\vA \ket{\psi_0}\bra{\psi_0} + \vA^{\perp}
\end{align}
for some operator $\vA^{\perp}$ such that $\vA^{\perp}\ket{\psi_0} = \bra{\psi_0}\vA^{\perp}= 0.$ 
\end{lem}
\begin{proof}
The off-diagonal terms vanish exactly
    \begin{align}
        (\vI-\ket{\psi_0}\bra{\psi_0})\hat{\vA}_f\ket{\psi_0}\bra{\psi_0} = \sum_{i=1}\ket{\psi_i}\bra{\psi_i}\sum_{\nu \in B(\vH), \nu \ge \Delta} \vA_{\nu} \hat{f}(-\nu)\ket{\psi_0}\bra{\psi_0} = 0,
    \end{align}    
    and similarly for $\ket{\psi_0}\bra{\psi_0}\hat{\vA}_f(\vI-\ket{\psi_0}\bra{\psi_0})$ since $\hat{\vA}_f$ is Hermitian.
\end{proof}

\section{Lower bounds} \label{sec:lower}
We collect two pieces of lower bounds.
\subsection{Hamiltonian evolution time}

Already for a single qubit, the Hamiltonian evolution time is optimal up to logarithmic factors in the black-box setting.
\begin{fact}[Lower bounds on black-box Hamiltonian evolution time] \label{fact:lowerbound_ham}
    Let $\vH$ be a $\Delta$-gapped Hamiltonian with ground state $\ket{\psi_0}$. Any algorithm that learns $\langle \psi_0|\vA\ket{\psi_0}$ to precision $\epsilon \ll 1$ with constant probability from a single copy of $\ket{\psi_0}$ and black-box controlled forwards and backwards Hamiltonian evolution must use $\Omega(1/(\Delta\epsilon))$ Hamiltonian evolution time.
\end{fact}
\begin{proof}
    Consider the single-qubit Hamiltonian $\vH = -(\cos(\theta)\vX + \sin(\theta)\vY)$, with ground state $\ket{\psi_0} = \frac{1}{\sqrt{2}} (|0\rangle + e^{i\theta}|1\rangle)$. Measuring $\bra{\psi_0}\vX\ket{\psi_0}$ and $\bra{\psi_0}\vY\ket{\psi_0}$ each to precision $\epsilon$ yields $\theta$ to precision $\mathcal{O}(\epsilon)$. This requires $\Omega(1/\epsilon)$ evolution time by standard metrology lower bounds (for example, see \cite[Theorem 2]{huang2023learning}). Moreover, the single copy of $\ket{\psi_0}$ cannot improve this scaling, since we can create a copy of $\ket{\psi_0}$ using constant Hamiltonian evolution time by projecting the maximally mixed state onto $\ket{\psi_0}$. Finally, by considering the Hamiltonian $\vH = -\Delta(\cos(\theta)\vX + \sin(\theta)\vY)$, we can see that the necessary evolution time is $\Omega(1/(\Delta\epsilon))$ for a $\Delta$-gapped Hamiltonian as desired.
\end{proof}

\subsection{Locality of filtered operator} \label{sec:locality_Heisenberg}

Following the original arguments of~\cite{hastings2006spectral}, we isolate a lower bound on the locality of a filtered operator that is closely related to our~$\hat{\vA}_f$ in terms of the correlation length of the ground state. We don't know how to give a similar lower bound on our $\hat{\vA}_f$, as the following decay of correlation argument exploits an asymmetric filter $\hat{f}$.

\begin{lem}[Quasi-local filtering implies decay of correlation]\label{lem:local_decay_correlation}
    Let $\vH$ be a $\Delta$-gapped Hamiltonian with ground state $\ket{\psi_0}$. Let $\vA,\vB$ be Hermitian operators such that $\norm{\vA},\norm{\vB}\le 1$, with spatially separated supports so that $[\vA,\vB]=0$. Consider the filtered operator $\hat{\vA}_f$ for a function $f$ such that $\hat{f}(-\nu) = 0$ when $\nu \ge \Delta$ and $\hat{f}(-\nu) = 1$ when $\nu \le 0$. If
    \begin{equation}
        \norm{[\vB,\hat{\vA}_f]} \le \eta,
    \end{equation}
    then the correlation is bounded by
    \begin{align}
        \labs{\bra{\psi_0} \vB \vA \ket{\psi_0}- \bra{\psi_0} \vB \ket{\psi_0} \bra{\psi_0}\vA\ket{\psi_0}} \le \eta,
    \end{align}
    as advertised.
\end{lem}
\begin{proof}
    Since $\hat{\vA}_{f}\ket{\psi_0} = (\bra{\psi_0}\vA\ket{\psi_0}) \cdot \ket{\psi_0}$, we have
    \begin{align}
        \bra{\psi_0} \vB \vA \ket{\psi_0}- \bra{\psi_0} \vB \ket{\psi_0} \bra{\psi_0}\vA\ket{\psi_0} &= \bra{\psi_0} \vB \vA \ket{\psi_0}- \bra{\psi_0} \vB \hat{\vA}_{f}\ket{\psi_0} \\
        &= \bra{\psi_0}(\vA -\hat{\vA}_{f}) \vB  \ket{\psi_0} + \bra{\psi_0} [\vB,\vA -\hat{\vA}_{f}] \ket{\psi_0}.
    \end{align}
    Now $\bra{\psi_0}(\vA -\hat{\vA}_{f}) = 0$ since there is no state with energies lower than the ground state, and
    \begin{equation}
        \labs{\bra{\psi_0} [\vB,\vA -\hat{\vA}_{f}] \ket{\psi_0}} \le \norm{[\vB,\hat{\vA}_{f}]} \le \eta.
    \end{equation}
\end{proof}

In the case of the Transverse Field Ising Model $(J=1)$
\begin{align}
\vH = \sum_{i} \vZ_i \vZ_{i+1} + g \sum_i \vX_i,    
\end{align}
the correlation length diverges with the inverse gap $\Omega(1/\Delta)$ as $g\rightarrow 1$ (see, e.g., ~\cite[Chapter 5]{Sachdev2011QPT}). Combining with~\autoref{lem:local_decay_correlation}, an $\Omega(1/\Delta)$ locality is required to implement the asymmetric filtered $\hat{\vA}_f$.

\subsection*{Acknowledgements}

CFC is supported by a Simons-CIQC postdoctoral fellowship through NSF QLCI Grant No. 2016245. The authors are grateful to the following individuals for insightful discussions: Hsin-Yuan Huang, Daniel Ranard, Thiago Bergamaschi, William Huggins, Stephen Jordan, Alec White, Guang Hao Low, Nicholas Rubin, Rolando Somma, Robin Kothari, Ryan Babbush, Thomas Vidick and Umesh Vazirani.

\bibliographystyle{alphaUrlePrint.bst}
\bibliography{ref,robbie_refs}

\appendix

\section{Implementation} \label{app:implementation}

Here, we describe the implementation details for \autoref{alg:main}, adapted from~\cite[Section III.B]{chen2023quantum}. We will assume black-box access to exact controlled Hamiltonian evolution:
\begin{align}
    \sum_{\bar{t}\in S_{t_0}} e^{\pm i \vH \bar{t}} \otimes \ket{\bar{t}}\bra{\bar{t}}
\end{align}
for some discretized time register of dimension $\labs{S_{t_0}}=2^k$
\begin{align}
    S_{t_0} = \L\{ -t_0 2^{k-1}, \cdots, -t_0,0,t_0, \cdots t_0 (2^{k-1}-1)\R\}.
\end{align}

\begin{lem}\label{lem:implement_Af}
Suppose $\norm{\vA} \le 1,$ there is an $\epsilon$-approximate block-encoding for $\hat{\vA}_f$ using one block-encoding for $\vA$ and 
\begin{align}
    \CO\L( \frac{\sqrt{\log(1/\epsilon)}}{\sigma}\R)\quad \text{or}\quad\CO\L( \frac{\log(1/\epsilon)\log\log(1/\epsilon)^3}{\Delta}\R) & \quad \text{controlled Hamiltonian evolution time for}\quad \vH 
\end{align}    
for the Gaussian with width $\sigma$~\eqref{eq:gaussian} or bump function with frequency support $[-\Delta,\Delta]$~\eqref{eq:compact_f}.
\end{lem}
\begin{proof}
Truncate the time integral and invoke a Riemann sum
\begin{equation}
    \int_{-\infty}^{\infty} \vA(t)f(t)\rd t = \int_{-T}^{T} \vA(t)f(t)\rd t+ \int_{\labs{t}> T} \vA(t)f(t)\rd t \approx  \sum_{\bar{t} \in S_{t_0}} \vA(\bar{t})f(\bar{t}) t_0.
\end{equation}

The Riemann sum error is bounded by (e.g., ~\cite[Lemma III.1]{chen2023efficient})
\begin{align}
\lnorm{ \sum_{\bar{t} \in S_{t_0}} \vA(\bar{t})f(\bar{t}) t_0 - \int_{-T}^{T} \vA(t)f(t)\rd t } \le \frac{2T^2}{\labs{S_{t_0}}} \L(\norm{[\vA,\vH]}\sup_t\labs{f(t)} + \sup_t\labs{f(t)'}\R).    
\end{align}
In the Gaussian case, $\sup_t\labs{f(t)} \le \sigma$, $\sup_t\labs{f(t)'}\le \CO(\sigma^2),$ and $\int_{\labs{t}> T} \labs{f(t)}\rd t \le \CO( e^{-\sigma^2T^2/2}/\sigma T )$. Optimize $T$ at
\begin{equation}
    T = \CO\L(\frac{\sqrt{\log(1/\epsilon)}}{\sigma}\R) \ , \quad |S_{t_0}| = \CO\L(\frac{2T^2\sigma^2}{\epsilon} (1+\norm{[\vA,\vH]}/\sigma)\R) 
\end{equation}
to obtain the advertised error bound.
For the bump function, $\sup_t\labs{f(t)} \le \CO(\Delta),\sup_t\labs{f(t)'}= \CO( \Delta^2)$ and $\int_{\labs{t}> T} \labs{f(t)}\rd t = \CO( \Delta T e^{-\frac{-2\Delta T}{7\ln(\Delta T)^2}}).$
\begin{equation}
    T = \CO\L(\frac{\log(1/\epsilon) \log\log(1/\epsilon)^3}{\Delta}\R) \ , \quad |S_{t_0}| = \CO\L(\frac{2T^2\Delta^2}{\epsilon} (1+\norm{[\vA,\vH]}/\Delta)\R) 
\end{equation}

The algorithm for the Riemann sum $ \sum_{\bar{t} \in S_{t_0}} \vA(\bar{t})f(\bar{t}) t_0$ is a direct LCU together with controlled Heisenberg evolution (\autoref{fig:A_LCU}), with state preperation circuit for $\sum_{\bar{t} \in S_{t_0}} \sqrt{f(\bar{t})t_0} \ket{\bar{t}}$ (\cite{McArdle_2022quantumstate,chen2023efficient}). 
\end{proof}

\begin{proof}[Proof of~\autoref{alg:main}]
    Perform high-confidence phase estimation (e.g.,~\cite[Section 13]{dalzell2023quantum} ) on $\hat{\vA}_f$ on $\ket{\psi_0}$ to obtain $\epsilon$-precision with $\delta/2$ failure probability, using $\CO( \log(1/\delta)/\epsilon) $-queries to $\hat{\vA}_f.$ For the Gaussian filter, we choose $\sigma = \Delta /\sqrt{\log(1/\delta')}$, and $T= \log(1/\delta')/\Delta$ for achieving $\delta'$ approximate block-encoding for $\hat{\vA}_f.$ Choose $\delta' = \CO(\delta\epsilon / \log(1/\delta))$ such that $\delta' \log(1/\delta)/\epsilon \le \delta$, to conclude the proof that $T=\CO( \log(1/\epsilon\delta')/\Delta )$ suffices for a good approximation for $\hat{\vA}_f$. The case of bump function is essentially the same scaling, up to $\poly \log\log$ factors. The resulting state can be slightly entangled with the ancilla due to the $\delta/2$-strength leakage, but tracing out the ancilla yields a state $\delta$-close in trace distance.
\end{proof}

\section{Locality analysis} \label{app:locality}

In this section, we analyze the locality of our measurement protocol to establish \autoref{alg:local}. For the lattice Hamiltonian $\vH$ and an operator $\vA$, define $\vH_r$ to be the local Hamiltonian involving only the terms in $\vH$ within radius $r$ of the support of $\vA$. Denote $\vA_{\vH}(t) = e^{i\vH t} \vA e^{-i\vH t}$ to be the usual Heisenberg evolution, and $\vA_{\vH_r}(t) = e^{i\vH_r t} \vA e^{-i\vH_r t}$ the Heisenberg evolution under the spatially-truncated Hamiltonian $\vH_r$. We recall the Lieb-Robinson bounds (e.g.,~\cite[Lemma 5]{haah2020quantum}, \cite{chen2023speed}).
\begin{lem}[Lieb-Robinson bounds]\label{lem:LRbounds}
    Let $\vH$ be a local Hamiltonian~\eqref{eq:localHam} and $\vA$ is a local observable with $\norm{\vA} \leq 1$ supported on $A\subset \Lambda$. Then, there is a constant $c$ depending only on the geometric parameters $R,q,D$ such that
    \begin{align}
        \lnorm{\vA_{\vH_r}(t) - \vA_{\vH}(t)} \le \labs{A}\frac{(ct)^{r}}{r!}.
    \end{align}
    where $A$ is the support of $\vA$.
\end{lem}

Now define
\begin{equation}
    \hat{\vA}_{f,r,T} := \frac{1}{\sqrt{2\pi}}\int_{-T}^{T} \vA_{\vH_r}(t) f(t) \rd t.
\end{equation}
The following lemma guarantees that the space-truncation forms a good approximation.
\begin{lem}[Space-truncation approximation] \label{lem:space_truncation_approx}
    In the setting of~\autoref{lem:LRbounds}, 
    \begin{equation}
        \|\hat{\vA}_{f,r,T} - \hat{\vA}_{f,T}\| \leq  \CO\L(\labs{A} \frac{(cT)^{r}}{r!}\R).
    \end{equation}
\end{lem}
\begin{proof} 
    We write
    \begin{equation}
        \lnorm{\hat{\vA}_{f,r,T} - \hat{\vA}_f} = \lnorm{\frac{1}{\sqrt{2\pi}}\int_{-T}^{T} \big(\vA_{\vH_r}(t) - \vA_{\vH}(t)\big) f(t) dt} \le \frac{1}{\sqrt{2\pi}} \labs{A} \frac{(ct)^{r}}{r!} \int_{-T}^{T} f(t)\rd t = \CO\L(\labs{A} \frac{(cT)^{r}}{r! }\R)
    \end{equation}
    The inequality holds by \autoref{lem:space_truncation_approx}, and the final bound since $\int_{-T}^{T} f(t)\rd t = \CO(1)$.
\end{proof}

\begin{proof}[Proof of \autoref{alg:local}]
    Recall, \autoref{alg:main} uses $T = \CO(\log(1/\delta) + \log(1/\epsilon))/\Delta$ truncation time for $\hat{\vA}_f$ for a final $\delta$ leakage and $1-\delta$ confidence. By \autoref{lem:space_truncation_approx}, the extra error incurred by a space truncation of radius $r$ is $\CO(\labs{A}(cT/r)^r)$ for some constant $c$. In order to ensure $\labs{A}(cT/r)^r = \CO(\delta)$, it suffices to set $r = e cT + \log(\labs{A}/\delta)$.
\end{proof}

\end{document}